\documentclass[twoside,11pt]{article}

\usepackage[preprint,nohyperref]{arxiv}

\usepackage{graphicx}
\usepackage{wrapfig}
\usepackage{float}
\usepackage{enumitem}
\usepackage[dvipsnames,table]{xcolor}
\usepackage{amssymb}
\usepackage{amsmath}
\usepackage{amsfonts}
\usepackage{mathtools}
\usepackage{bm}
\usepackage{bbm}
\usepackage{multirow}
\usepackage{booktabs}
\usepackage{longtable}
\usepackage{microtype}
\usepackage{subcaption}
\usepackage{tikz}
\usetikzlibrary{arrows.meta}
\DisableLigatures[f]{encoding = *, family = *}
\usepackage{nicefrac}
\usepackage[utf8]{inputenc}
\usepackage[T1]{fontenc}
\usepackage{natbib}

\usepackage{algorithm}
\usepackage{algorithmic}
\usepackage{xpatch}

\usepackage[
colorlinks=true,
urlcolor=RoyalBlue,
citecolor=Green,
anchorcolor=red,
backref=page
]{hyperref}
\usepackage{url}
\usepackage{xurl}

\renewcommand{\algorithmicrequire}{\textbf{Input:}}
\renewcommand{\algorithmicensure}{\textbf{Output:}}
\makeatletter
\def\theHALC@line{\arabic{algorithm}.\arabic{ALC@line}}
\xpatchcmd{\algorithmic}
  {\addtocounter{ALC@line}{1}}
  {\refstepcounter{ALC@line}}
  {}
  {\PackageError{main}{Cannot enable algorithm line references}{}}
\makeatother

\newcommand{\algcomment}[1]{\hfill$\triangleright$~#1}

\xapptocmd{\thebibliography}
  {\interlinepenalty=10000\relax}
  {}
  {\PackageError{main}{Cannot keep references on one page}{}}

\usepackage[capitalize,noabbrev]{cleveref}
\crefname{section}{Section}{Sections}
\crefname{subsection}{Section}{Sections}
\crefname{appendix}{Appendix}{Appendices}
\crefname{figure}{Figure}{Figures}
\crefname{table}{Table}{Tables}
\crefname{algorithm}{Algorithm}{Algorithms}
\crefname{theorem}{Theorem}{Theorems}

\crefformat{equation}{Eq.~(#2#1#3)}
\crefrangeformat{equation}{Eqs.~(#3#1#4)~to~(#5#2#6)}
\crefmultiformat{equation}{Eqs.~(#2#1#3)}{,~(#2#1#3)}{, (#2#1#3)}{,~(#2#1#3)}
\crefrangemultiformat{equation}{Eqs.~(#3#1#4)~to~(#5#2#6)}{,~(#3#1#4)~to~(#5#2#6)}{, (#3#1#4)~to~(#5#2#6)}{,~(#3#1#4)~to~(#5#2#6)}

\newcommand{\methodname}{GPTQ-2D}
\newcommand{\method}{\texorpdfstring{\mbox{\methodname}}{\methodname}}

\firstpageno{1}

\begin{document}

\title{\method{}: Cubic-Time Two-Sided Adaptive Rounding}

\author{\name Jiale Chen
\email Jiale.Chen@ist.ac.at
\hfill
\addr Institute of Science and Technology Austria (ISTA)
\AND
\name Torsten Hoefler
\hfill
\addr ETH Z\"urich
\AND
\name Dan Alistarh\thanks{This author was affiliated with Red Hat AI during part of this work.}
\hfill
\addr Institute of Science and Technology Austria (ISTA)
}

\maketitle

\begin{abstract}
Adaptive rounding methods such as GPTQ quantize the weights of a large language model (LLM) one at a time, propagating each rounding error to the weights not yet rounded.
YAQA lowers the KL divergence from the original model with two-sided adaptive rounding under a Kronecker-factored Hessian, but its error feedback takes quartic work in the size of a square layer.
We present \method{}, which provably computes the same rounding with cubic work, while still rounding the entries of each anti-diagonal in parallel.
Its GPU implementation uses two levels of blocking to improve hardware utilization.
We also extend the algorithm to Hessians with any number of Kronecker factors.
In our experiments, \method{} is as accurate as YAQA and cuts the time to quantize Llama-3.1-70B-Instruct on a single NVIDIA RTX A6000 GPU from 8.6 days to 1.6 hours.
Source code is available at \url{https://github.com/IST-DASLab/GPTQ-2D}.
\end{abstract}

\section{Introduction}
\label{sec:intro}

Post-training quantization rounds the weights of a trained large language model (LLM) to a few bits.
Let $X\in\mathbb{R}^{m\times n}$ be the weights of a linear layer with $m$ output and $n$ input channels.
The widely used GPTQ~\citep{frantar2023optq} rounds $X$ one column at a time.
It feeds the rounding error of each column back to the columns not yet rounded.
It aims to keep the objective $\left\|\left(Z-X\right)B\right\|_{\mathrm{F}}^2$ small for the rounded matrix $Z\in\mathbb{Z}^{m\times n}$.
Here $BB^{\top}$ is the second moment of the layer's input activations.

GPTQ's objective is one-sided.
It weighs each weight error by the layer's inputs, but ignores how later layers propagate it.
YAQA~\citep{pmlr-v306-tseng26b} instead targets the KL divergence from the original model.
The Hessian of this KL divergence with respect to a layer's weights describes its local curvature.
This $mn\times mn$ Hessian matrix is too large to use directly.
YAQA approximates it by a Kronecker product.
Its factors are $G\in\mathbb{R}^{m\times m}$ on the output side and $H\in\mathbb{R}^{n\times n}$ on the input side.
YAQA rounds under the two-sided objective
\begin{equation}
\left\|A\left(Z-X\right)B\right\|_{\mathrm{F}}^2,
\qquad
A^{\top}A=G,
\quad
BB^{\top}=H,
\label{eq:intro-objective}
\end{equation}
which reduces to GPTQ's objective for $A=I$.
YAQA lowers the KL divergence compared with GPTQ.
\cref{fig:layer-error} shows the two-sided objective value of \cref{eq:intro-objective} for individual layers.
In this plot, GPTQ uses only one of the two factors for rounding, but is evaluated on the same two-sided objective.
\method{}, our faster method for YAQA's rounding, uses both and achieves a lower objective value in every layer.

\begin{figure}[!htbp]
\begin{center}
\includegraphics{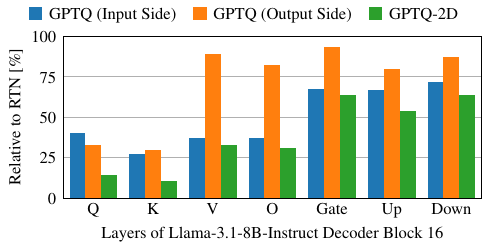}
\end{center}
\caption{Two-sided objective values (\cref{eq:intro-objective}) for seven LLM layers, relative to round to nearest (RTN). All methods use YAQA's precomputed Hessian factors and round to INT4 with the same static BF16 scales per 32 entries.}
\label{fig:layer-error}
\end{figure}

Two-sided adaptive rounding is more costly than one-sided.
Vectorizing $X$ turns \cref{eq:intro-objective} into a one-sided objective over all $mn$ weights.
GPTQ's algorithm then applies, but its error feedback alone takes $O\left(m^2n^2\right)$ work, which is quartic for a square layer.
YAQA needs only $m+n-1$ sequential stages by exploiting the parallel structure of the problem, but its total feedback work is still quartic.

A model is often deployed in multiple formats and bitwidths.
Unlike the Hessians, which are computed once and can be shared, the rounding is repeated for each format and bitwidth.
\method{} reduces YAQA's rounding work from quartic complexity to cubic, while keeping its $m+n-1$ sequential stages.
We extend this result to $D$\nobreakdash-sided rounding with $D>2$.
We also provide efficient blocked GPU implementations.
In QTIP 2-bit, \method{} matches YAQA's accuracy on Llama-3.1-8B-Instruct and quantizes Llama-3.1-70B-Instruct from precomputed Hessians in 1.6 hours instead of YAQA's 8.6 days on an NVIDIA RTX A6000 GPU.

\section{Related Work}
\label{sec:related}

\textbf{Adaptive rounding via second-order information.}
Pruning and quantization guided by second-order information build on Optimal Brain Surgeon~\citep{hassibi1993obs}, which ranks weights by inverse-Hessian saliencies and applies a closed-form error-feedback update to the remaining weights.
Optimal Brain Compression~\citep{NEURIPS2022_1caf09c9} applies the same one-at-a-time scheme to the layer-wise post-training objective, in which every row has the input second moment as its Hessian.
GPTQ~\citep{frantar2023optq} fixes one column order for all rows.
The inverse Hessian then needs only one LDL decomposition,\footnote{Like \citet{ICLR2026_c76edb90}, we simplify GPTQ in this paper by replacing its Cholesky factorization with an LDL decomposition.} which reduces the cost from quartic to cubic for a square layer and makes LLM-scale quantization practical.
QuIP~\citep{NEURIPS2023_0df38cd1} recasts this triangular error feedback as LDLQ and proves that LDLQ is equivalent to GPTQ.
These layer-wise methods are all \emph{one-sided}, with a single Hessian acting on the input side of the weight matrix.

\textbf{Lattice and nearest plane view.}
Adaptive rounding under a quadratic objective is exactly Babai's nearest plane algorithm~\citep{babai1986} for the closest vector problem, on the lattice induced by the objective.
\citet{ICLR2026_c76edb90} and \citet{ICLR2026_a6a6891c} made this connection precise for GPTQ.
The lattice view also yields an upper bound on GPTQ's objective value, but no optimality guarantee.
Finding the closest lattice point is NP-hard in general~\citep{dinur2003}.

\textbf{Kronecker-factored two-sided rounding.}
YAQA~\citep{pmlr-v306-tseng26b} proposes an LLM quantization framework based on two-sided adaptive rounding.
It rounds each layer under a Kronecker-factored Hessian, generalizing the fixed-point iteration of LDLQ to a left and a right factor.
Exploiting the Kronecker structure, YAQA bounds the number of stages by $m+n-1$, and its reference implementation sweeps the weight matrix by anti-diagonals.
Each stage applies the feedback as dense matrix products, so the total feedback work is $O\left(mn\left(\max\left(m,n\right)\right)^2\right)$, which is quartic for a square matrix.
\method{} needs only $O\left(mn\max\left(m,n\right)\right)$ feedback work, which is cubic for a square matrix.
YAQA rounds in reverse order, from the last row and column, whereas \method{} rounds forward.
Run in YAQA's order, \method{} computes the same rounding in exact arithmetic.
It can thus replace YAQA's rounding pass and reuse the same Hessian factors.
BoA~\citep{pmlr-v267-kim25j} performs nested GPTQ under Kronecker-factored Hessians.
Its released implementation takes $O\left(m^2n^2\right)$ feedback work.
Reassociating its matrix products would reduce this to $O\left(mn\max\left(m,n\right)\right)$, but BoA still needs $O\left(mn\right)$ stages.
TurboBoA~\citep{ICLR2026_dd3bd4e3} accelerates BoA by rounding several output channels at once, without feedback between them.
Unlike TurboBoA, \method{} preserves BoA's rounding in exact arithmetic, with cubic work and $m+n-1$ stages.
Grouped output-channel quantization is also supported through super-entries (\cref{sec:quantizers}).

\method{} was made publicly available (e.g., via a preprint) before the concurrent work of FastKron~\citep{birnick2026fastkronefficientquantizationkroneckerfactored}, which also needs only cubic work and $O\left(\max\left(m,n\right)\right)$ stages.
\method{} is iterative, whereas FastKron is recursive.
\method{} also supports super-entries.
It computes the corrected values of each super-entry exactly before rounding it.
This makes room for quantizers whose scales or transforms are computed from these values, such as WUSH~\citep{pmlr-v306-chen26z}.
We additionally extend \method{} to $D$\nobreakdash-sided rounding with $D>2$.

\section{Preliminaries}
\label{sec:prelim}

This section recalls one-dimensional adaptive rounding (\cref{sec:1d}), states the two-sided problem (\cref{sec:problem}), and extends the bound on GPTQ's objective value to the two-sided case (\cref{sec:error-bound}).
It then defines the quantizers and super-entries that our results cover (\cref{sec:quantizers}).

\subsection{One-Dimensional Adaptive Rounding}
\label{sec:1d}

\begin{algorithm}[!htbp]
\caption{One-dimensional adaptive rounding}
\label{alg:1d}
\textbf{Input:} input vector $x \in \mathbb{R}^{N}$, unit lower triangular feedback $F \in \mathbb{R}^{N \times N}$

\textbf{Output:} rounded vector $z \in \mathbb{Z}^{N}$
\begin{algorithmic}[1]
\STATE Initialize $z \in \mathbb{Z}^{N}$
\STATE $y \gets x$
\FOR{$k = 1, \ldots, N$}
    \STATE $z_k \gets \left\lfloor y_k \right\rceil$
    \STATE $e_k \gets z_k - y_k$
    \STATE $y_{k:N} \gets y_{k:N} + F_{k:N,\,k} e_k$
\ENDFOR
\STATE \textbf{return} $z$
\end{algorithmic}
\end{algorithm}

We first recall the standard one-dimensional adaptive rounding primitive.
Let $x\in\mathbb{R}^N$ be an input vector, $\left\lfloor\cdot\right\rceil$ the nearest-integer rounding operator, and $F\in\mathbb{R}^{N\times N}$ a unit lower triangular feedback matrix.
A unit triangular matrix is a triangular matrix with ones on its diagonal.
The feedback matrix is fixed by the objective.
GPTQ rounds $x$ to an integer vector $z\in\mathbb{Z}^N$ under the one-sided objective $\left\|\left(z-x\right)^{\top}\mathcal{B}\right\|^2$, with a nonsingular basis matrix $\mathcal{B}\in\mathbb{R}^{N\times N}$ and the Gram matrix $\mathcal{B}\mathcal{B}^{\top}$.
It takes $F$ from the LDL decomposition of the inverse Gram matrix, $F\Lambda_F F^{\top}=\left(\mathcal{B}\mathcal{B}^{\top}\right)^{-1}$, with $\Lambda_F\in\mathbb{R}^{N\times N}$ being diagonal~\citep{ICLR2026_c76edb90}.
We will later redefine $F$ for the two-sided objective in \cref{eq:vectorized-feedback}.
Adaptive rounding maintains a corrected vector $y\in\mathbb{R}^N$, initialized to $y=x$, and sweeps the coordinates in a fixed order.
At coordinate $k$ ($k = 1, \ldots, N$) it rounds $z_k = \left\lfloor y_k \right\rceil$ and records the local rounding error $e_k = z_k-y_k$.
It then propagates the error to the remaining coordinates through the $k$-th column of $F$, with the update $y \leftarrow y + F_{:,\,k} e_k$.
Because $F$ is unit lower triangular, this update affects only coordinates $k,k+1,\ldots,N$, so a single forward sweep suffices (\cref{alg:1d}), at a cost of $O\left(N^2\right)$ arithmetic operations.
Equivalently, if $e\in\mathbb{R}^N$ collects the errors of the already-processed coordinates and is zero elsewhere, the implicit corrected vector is
\begin{equation}
y=x+Fe.
\label{eq:1d-feedback}
\end{equation}

\subsection{Two-Sided Matrix Rounding}
\label{sec:problem}

We now pose the two-sided problem.
Let $X\in\mathbb{R}^{m\times n}$ be an input matrix, and let $A\in\mathbb{R}^{m\times m}$ and $B\in\mathbb{R}^{n\times n}$ be nonsingular basis matrices.
We round $X$ to an integer matrix $Z\in\mathbb{Z}^{m\times n}$ under the two-sided objective.
We write the objective in terms of the residual $R=Z-X\in\mathbb{R}^{m\times n}$ as
\begin{equation}
\left\|A\left(Z-X\right)B\right\|_{\mathrm{F}}^2
=
\left\|ARB\right\|_{\mathrm{F}}^2.
\label{eq:objective}
\end{equation}
Define the Gram matrices $G\in\mathbb{R}^{m\times m}$ and $H\in\mathbb{R}^{n\times n}$ as
\begin{equation}
G=A^{\top}A,
\qquad
H=BB^{\top}.
\label{eq:GH}
\end{equation}
Let $\otimes$ denote the Kronecker product.
The column-major vectorization $\operatorname{vec}\left(X\right)\in\mathbb{R}^{mn}$ concatenates the columns of $X$, so the entry $X_{i,\,j}$ becomes its $\left(\left(j-1\right)m+i\right)$-th coordinate.
This vectorization satisfies
\begin{equation}
\operatorname{vec}\left(ARB\right)=\left(B^{\top}\otimes A\right)\operatorname{vec}\left(R\right).
\label{eq:vec}
\end{equation}
The objective is a quadratic form in $\operatorname{vec}\left(R\right)$,
\begin{equation}
\begin{split}
\left\|ARB\right\|_{\mathrm{F}}^2
&=\left\|\left(B^{\top}\otimes A\right)\operatorname{vec}\left(R\right)\right\|^2\\
&=\operatorname{vec}\left(R\right)^{\top}\left(B\otimes A^{\top}\right)\left(B^{\top}\otimes A\right)\operatorname{vec}\left(R\right)\\
&=\operatorname{vec}\left(R\right)^{\top}\left(\left(BB^{\top}\right)\otimes\left(A^{\top}A\right)\right)\operatorname{vec}\left(R\right)\\
&=\operatorname{vec}\left(R\right)^{\top}\left(H\otimes G\right)\operatorname{vec}\left(R\right).
\end{split}
\label{eq:kron-objective}
\end{equation}
Thus two-sided matrix rounding is exactly one-dimensional rounding of $\operatorname{vec}\left(X\right)$ under the Kronecker Gram matrix $H\otimes G$, and the one-dimensional primitive of \cref{sec:1d} applies with $N=mn$.
The inverses of the Gram matrices of \cref{eq:GH} have the LDL decompositions
\begin{equation}
L\Lambda_L L^{\top}=G^{-1},
\qquad
U^{\top}\Lambda_U U=H^{-1},
\label{eq:LU}
\end{equation}
with unit lower triangular $L\in\mathbb{R}^{m\times m}$, unit upper triangular $U\in\mathbb{R}^{n\times n}$, and diagonal $\Lambda_L\in\mathbb{R}^{m\times m}$ and $\Lambda_U\in\mathbb{R}^{n\times n}$.
We call $L$ and $U$ the feedback factors.
The inverse of the Kronecker Gram matrix factors as
\begin{equation}
\begin{split}
\left(H\otimes G\right)^{-1}
&=H^{-1}\otimes G^{-1}\\
&=\left(U^{\top}\Lambda_U U\right)\otimes\left(L\Lambda_L L^{\top}\right)\\
&=\left(U^{\top}\otimes L\right)\left(\Lambda_U\otimes\Lambda_L\right)\left(U^{\top}\otimes L\right)^{\top}.
\end{split}
\label{eq:kron-ldl}
\end{equation}
Since $U^{\top}\otimes L$ is unit lower triangular and $\Lambda_U\otimes\Lambda_L$ is diagonal, this is the LDL decomposition of $\left(H\otimes G\right)^{-1}$.
Hence the feedback matrix $F\in\mathbb{R}^{mn\times mn}$ for $\operatorname{vec}\left(X\right)$ is
\begin{equation}
F = U^{\top}\otimes L.
\label{eq:vectorized-feedback}
\end{equation}
The remainder of the paper develops faster algorithms and efficient implementations that exploit the Kronecker structure to avoid the $O\left(m^2n^2\right)$ cost of the vectorized sweep.

\subsection{Bound on the Objective}
\label{sec:error-bound}

In the lattice view of \cref{sec:related}, one-dimensional rounding of $\operatorname{vec}\left(X\right)$ with $F=U^{\top}\otimes L$ is Babai's nearest plane algorithm on the lattice with basis $B^{\top}\otimes A$.
The upper bound on the objective value of one-dimensional GPTQ without clipping~\citep{ICLR2026_c76edb90} therefore holds for two-sided rounding as well.

\subsection{Quantizers and Super-Entries}
\label{sec:quantizers}

The main results of this paper hold analogously for any \emph{local} quantizer in place of $\left\lfloor\cdot\right\rceil$.
A quantizer is local if it computes the rounded value of an entry deterministically from its position and its (corrected) value.
A local grouped or vector quantizer, such as the grouped integer quantizer or the trellis quantizer of QTIP~\citep{NEURIPS2024_6de2e84b}, jointly rounds a $p\times q$ sub-matrix.
We call the sub-matrix a super-entry.
The algorithms and results then apply to $X$ as an $\left(m/p\right)\times\left(n/q\right)$ matrix of super-entries.
Here $L$ and $U$ come from the LDL decompositions of \cref{eq:LU} over $p\times p$ and $q\times q$ super-entries, so their diagonal super-entries are identity matrices.
The scalar products in the algorithms become small matrix products, in the order written.

\section{Theoretical Results}
\label{sec:theory}

We now lift the one-dimensional sweep of \cref{sec:1d} to the two-sided problem.
\cref{sec:feedback,sec:parallel} derive \method{}, a cubic-time algorithm that provably reproduces the vectorized trajectory.
\cref{sec:variants} compares the costs of its variants, and \cref{sec:d-sided} extends it to $D$\nobreakdash-sided rounding.

\subsection{Two-Sided Feedback and the Rank-One Update}
\label{sec:feedback}

Let $E\in\mathbb{R}^{m\times n}$ collect the local rounding errors.
By \cref{eq:1d-feedback,eq:vectorized-feedback}, the implicit corrected matrix $Y\in\mathbb{R}^{m\times n}$ satisfies
\begin{equation}
\begin{split}
\operatorname{vec}\left(Y\right)
&=\operatorname{vec}\left(X\right)+F\operatorname{vec}\left(E\right)\\
&=\operatorname{vec}\left(X\right)+\left(U^{\top}\otimes L\right)\operatorname{vec}\left(E\right)\\
&=\operatorname{vec}\left(X+LEU\right).
\end{split}
\label{eq:Y-LEU}
\end{equation}
Hence $Y=X+LEU$.
Suppose the algorithm rounds entry $\left(i,j\right)$ and produces error $E_{i,\,j}$.
Let $J^{i,\,j}\in\left\{0,1\right\}^{m\times n}$ be the matrix with a single one at entry $\left(i,j\right)$ and zeros elsewhere.
The contribution of this error to the corrected matrix is $L\left(E_{i,\,j}J^{i,\,j}\right)U=L_{:,\,i}E_{i,\,j}U_{j,\,:}$.
A scalar rounding error at $\left(i,j\right)$ thus produces the rank-one feedback update
\begin{equation}
Y \leftarrow Y + L_{:,\,i}E_{i,\,j}U_{j,\,:}.
\label{eq:rank-one-update}
\end{equation}
This is the central two-dimensional feedback formula.

\begin{figure}[!htbp]
\begin{center}
\begin{tikzpicture}[x=1.4cm,y=1.0cm,font=\scriptsize,line width=0.5pt,>=Straight Barb]
  \foreach \i/\j in {1/1,1/2,1/4,2/1,2/3,2/4,3/2,3/3,3/4}{
    \node (n\i\j) at (\j,-\i) {$X_{\i,\,\j}$};
  }
  \foreach \i/\j in {1/3,2/2,3/1}{
    \node[red] (n\i\j) at (\j,-\i) {$X_{\i,\,\j}$};
  }
  \foreach \i in {1,2,3}{
    \foreach \j/\jj in {1/2,2/3,3/4}{
      \draw[->] (n\i\j) -- (n\i\jj);
    }
  }
  \foreach \i/\ii in {1/2,2/3}{
    \foreach \j in {1,2,3,4}{
      \draw[->] (n\i\j) -- (n\ii\j);
    }
  }
\end{tikzpicture}
\end{center}
\caption{The dependency graph of the two-sided sweep on a $3\times4$ matrix, with arrows marking immediate dependencies.
Entries on a common anti-diagonal, such as $\mathcal{I}_4$ (\textcolor{red}{red}), can be rounded in parallel.}
\label{fig:dependency}
\end{figure}

Because $L$ is lower triangular, the column $L_{:,\,i}$ is supported only on rows $i,\ldots,m$.
Because $U$ is upper triangular, the row $U_{j,\,:}$ is supported only on columns $j,\ldots,n$.
Hence the rank-one update of \cref{eq:rank-one-update} is supported only on the lower-right rectangle $\left\{i,\ldots,m\right\}\times\left\{j,\ldots,n\right\}$.
An update from $\left(i',j'\right)$ can reach $\left(i,j\right)$ only if $i'\le i$ and $j'\le j$.
The rounding order only needs to place each entry after all entries that send it feedback.
Such an order is a topological order of the dependency graph (\cref{fig:dependency}).
All such orders give the same corrected values at rounding and the same rounded matrix.
We group the entries by anti-diagonal $s=2,\ldots,m+n$ into
\begin{equation}
\mathcal{I}_s
=
\left\{\left(i,j\right)\;\middle|\;1\le i\le m,\ 1\le j\le n,\ i+j=s\right\}.
\label{eq:anti-diagonal}
\end{equation}
For entries on the same anti-diagonal, the feedback conditions $i'\le i$ and $j'\le j$ hold only if the entries coincide.
Thus entries on a common anti-diagonal are independent.
They may be rounded in parallel once all previous anti-diagonals have been processed.
We call the rounding of one anti-diagonal, with the propagation of its errors, a stage.
The anti-diagonal sweep thus takes $m+n-1$ stages.

\begin{algorithm}[!htbp]
\caption{Direct anti-diagonal dense updates}
\label{alg:dense}
\textbf{Input:} input matrix $X \in \mathbb{R}^{m\times n}$, feedback factors $L \in \mathbb{R}^{m \times m}$ and $U \in \mathbb{R}^{n \times n}$

\textbf{Output:} rounded matrix $Z \in \mathbb{Z}^{m\times n}$
\begin{algorithmic}[1]
\STATE Initialize $Z \in \mathbb{Z}^{m\times n}$
\STATE $Y \gets X$
\FOR{$s = 2, \ldots, m + n$}
    \FORALL{$\left(i, j\right) \in \mathcal{I}_s$ \textbf{in parallel}}
        \STATE $Z_{i,\,j} \gets \left\lfloor Y_{i,\,j}\right\rceil$
        \STATE $E_{i,\,j} \gets Z_{i,\,j} - Y_{i,\,j}$
    \ENDFOR
    \STATE $Y \gets Y + \sum_{\left(i, j\right) \in \mathcal{I}_s} L_{:,\,i}\, E_{i,\,j}\, U_{j,\,:}$
\ENDFOR
\STATE \textbf{return} $Z$
\end{algorithmic}
\end{algorithm}

The most direct realization of the anti-diagonal sweep keeps the corrected matrix $Y=X+LEU$ explicitly (\cref{alg:dense}).
It rounds each anti-diagonal in parallel and then applies the rank-one update of \cref{eq:rank-one-update} for all of its entries.
The additive updates commute, so the updates of each anti-diagonal fuse into one dense update.
It takes only $m+n-1$ stages.
However, each of the $mn$ rank-one updates touches up to $mn$ entries, so the total work is $O\left(m^2n^2\right)$.

\subsection{\method{} Algorithm}
\label{sec:parallel}

\begin{figure}[!htbp]
\begin{center}
\includegraphics[width=\linewidth]{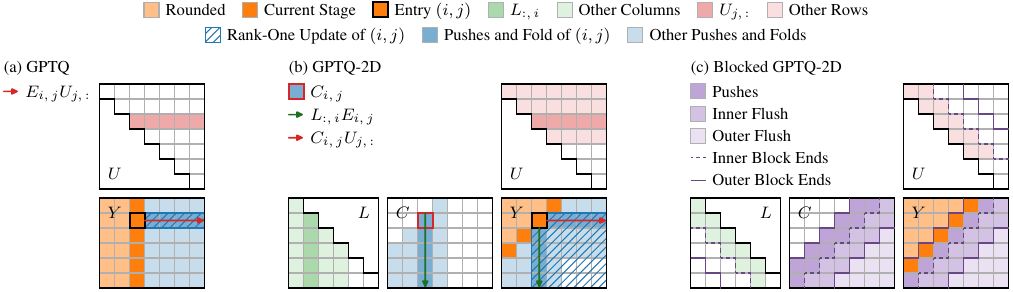}
\end{center}
\caption{One stage of (a) GPTQ, (b) \method{} (\cref{alg:parallel}) and (c) blocked \method{} (\cref{alg:blocked}) on a $6\times7$ matrix. GPTQ is the one-sided case $L=I$, so its buffer is $C=E$ and its stage is a column. In (b), entry $\left(i,j\right)$ folds $L_{:,\,i}E_{i,\,j}$ into $C_{:,\,j}$, pushes the same product down into $Y_{:,\,j}$, and pushes $C_{i,\,j}U_{j,\,:}$ right into $Y_{i,\,:}$. The rest of its rank-one update reaches $Y_{i+1:m,\,j+1:n}$ through $C_{i+1:m,\,j}$ as the entries below it push right. In (c), the folds and pushes stop at the end of the inner block of 2 stages. An inner flush then updates the rest of the outer block of 4 stages, and an outer flush updates all later stages. Purple marks where the inner block's feedback lands, darker the sooner.}
\label{fig:algorithm}
\end{figure}

\begin{algorithm}[!htbp]
\caption{\method{}}
\label{alg:parallel}
\textbf{Input:} input matrix $X \in \mathbb{R}^{m\times n}$, feedback factors $L \in \mathbb{R}^{m \times m}$ and $U \in \mathbb{R}^{n \times n}$

\textbf{Output:} rounded matrix $Z \in \mathbb{Z}^{m\times n}$
\begin{algorithmic}[1]
\STATE Initialize $Z \in \mathbb{Z}^{m\times n}$
\STATE $Y \gets X$
\STATE $C \gets 0 \in \mathbb{R}^{m\times n}$ \algcomment{buffer $C=LE$}
\FOR{$s = 2, \ldots, m + n$}
    \FORALL{$\left(i, j\right) \in \mathcal{I}_s$ \textbf{in parallel}}
        \STATE $Z_{i,\,j} \gets \left\lfloor Y_{i,\,j}\right\rceil$ \label{line:round}
        \STATE $E_{i,\,j} \gets Z_{i,\,j} - Y_{i,\,j}$ \label{line:err}
        \STATE $M_{i:m,\,j} \gets L_{i:m,\,i}\, E_{i,\,j}$
        \STATE $C_{i:m,\,j} \gets C_{i:m,\,j} + M_{i:m,\,j}$ \label{line:cupdate} \algcomment{fold}
        \STATE $Y_{i+1:m,\,j} \gets Y_{i+1:m,\,j} + M_{i+1:m,\,j}$ \label{line:pushdownY} \algcomment{push down}
    \ENDFOR
    \FORALL{$\left(i, j\right) \in \mathcal{I}_s$ \textbf{in parallel}}
        \STATE $Y_{i,\,j+1:n} \gets Y_{i,\,j+1:n} + C_{i,\,j}U_{j,\,j+1:n}$ \label{line:pushright} \algcomment{push right}
    \ENDFOR
\ENDFOR
\STATE \textbf{return} $Z$
\end{algorithmic}
\end{algorithm}

\method{} (\cref{alg:parallel}) is our main algorithm.
It produces the same rounded matrix as \cref{alg:1d,alg:dense} (\cref{thm:equiv}) but in cubic instead of quartic work.
Like \cref{alg:dense}, it rounds one anti-diagonal at a time, maintaining the corrected values of the unrounded entries in $Y$.
It additionally carries a buffer $C=LE\in\mathbb{R}^{m\times n}$.
It propagates each rounding error only along its own column (through $L$) and row (through $U$), rather than over the entire lower-right rectangle.
\cref{fig:algorithm} contrasts one stage of \method{} with one of GPTQ.
There also exists a mirror algorithm that maintains the buffer $EU$ instead of $LE$.

In \cref{alg:parallel}, each entry does an $O\left(1\right)$ rounding and $O\left(\max\left(m,n\right)\right)$ feedback.
Its fold of length $m-i+1$ into $C$ and its downward push of length $m-i$ reuse one product $M_{i:m,\,j}=L_{i:m,\,i}E_{i,\,j}$.
Its rightward push has length $n-j$.
The sweep thus costs $O\left(mn\max\left(m,n\right)\right)$ work in $m+n-1$ stages.
The one-time computation of $L$ and $U$ inverts the two Gram matrices and LDL-factorizes the inverses.
It costs $O\left(\left(\max\left(m,n\right)\right)^3\right)$ and dominates, so the end-to-end cost is $O\left(\left(\max\left(m,n\right)\right)^3\right)$.
The working storage is $O\left(\left(\max\left(m,n\right)\right)^2\right)$.

\begin{theorem}
\label{thm:equiv}
\cref{alg:parallel} produces the same rounded values as \cref{alg:1d} with $x=\operatorname{vec}\left(X\right)$ and $F=U^{\top}\otimes L$.
\end{theorem}
\begin{proof}
By \cref{sec:feedback,eq:Y-LEU}, it suffices to show that \cref{alg:parallel} rounds each entry $\left(i,j\right)$ at $X_{i,\,j}+\left(LEU\right)_{i,\,j}$.
Here $E$ contains the errors (Line~\ref{line:err}) from earlier stages and is zero elsewhere.

In \cref{alg:parallel}, $Y_{i,\,j}$ starts at $X_{i,\,j}$.
Before $\left(i,j\right)$ is rounded, it receives the downward pushes (Line~\ref{line:pushdownY}) of the entries above it and the rightward pushes (Line~\ref{line:pushright}) of those left of it, so
\begin{equation}
Y_{i,\,j}=X_{i,\,j}+\sum_{i'<i}L_{i,\,i'}E_{i',\,j}+\sum_{j'<j}C_{i,\,j'}U_{j',\,j}.
\label{eq:parallel-pushes}
\end{equation}
Since $E_{i,\,j}=0$ before $\left(i,j\right)$ is rounded and $L$ is lower triangular, the downward pushes add up to
\begin{equation}
\sum_{i'<i}L_{i,\,i'}E_{i',\,j}=\sum_{i'\le i}L_{i,\,i'}E_{i',\,j}=\left(LE\right)_{i,\,j}.
\label{eq:parallel-down}
\end{equation}
The rightward push of each entry $\left(i,j'\right)$ comes after all entries $\left(i',j'\right)$ with $i'\le i$ have folded into $C_{i,\,j'}$ (Line~\ref{line:cupdate}).
So it reads
\begin{equation}
C_{i,\,j'}=\sum_{i'\le i}L_{i,\,i'}E_{i',\,j'}=\left(LE\right)_{i,\,j'}.
\label{eq:parallel-right}
\end{equation}
Substituting \cref{eq:parallel-down,eq:parallel-right} into \cref{eq:parallel-pushes}, with $U$ unit upper triangular, gives
\begin{equation}
\begin{split}
Y_{i,\,j}
&=X_{i,\,j}+\left(LE\right)_{i,\,j}+\sum_{j'<j}\left(LE\right)_{i,\,j'}U_{j',\,j}\\
&=X_{i,\,j}+\left(LE\right)_{i,\,j}U_{j,\,j}+\sum_{j'<j}\left(LE\right)_{i,\,j'}U_{j',\,j}\\
&=X_{i,\,j}+\sum_{j'\le j}\left(LE\right)_{i,\,j'}U_{j',\,j}\\
&=X_{i,\,j}+\left(LEU\right)_{i,\,j},
\end{split}
\end{equation}
as required.
\end{proof}

\subsection{Other Variants}
\label{sec:variants}

\begin{table}[!htbp]
\caption{Cost of computing the rounding trajectory. Setup forms the LDL factors of Gram inverses, sweep is the rounding pass, and depth counts sequential stages. GPTQ is listed for reference. For $n\ge m$, all three costs of \method{} match those of GPTQ, so the two-sided extension is asymptotically free.}
\label{tab:complexity}
\begin{center}
\begin{scriptsize}
\setlength{\tabcolsep}{4pt}
\begin{tabular}{lllll}
\toprule
Method & Setup & Sweep & Depth & Algorithm \\
\midrule
\multicolumn{5}{l}{\emph{One-sided objective} $\left\|\left(Z-X\right)B\right\|_{\mathrm{F}}^2$} \\
GPTQ & $O\left(n^3\right)$ & $O\left(mn^2\right)$ & $O\left(n\right)$ & \ref{alg:1d} \\
\midrule
\multicolumn{5}{l}{\emph{Two-sided objective} $\left\|A\left(Z-X\right)B\right\|_{\mathrm{F}}^2$} \\
Vectorized adaptive rounding & $O\left(\left(\max\left(m,n\right)\right)^3\right)$ & $O\left(m^2n^2\right)$ & $O\left(mn\right)$ & \ref{alg:1d} \\
Direct anti-diagonal dense updates & $O\left(\left(\max\left(m,n\right)\right)^3\right)$ & $O\left(m^2n^2\right)$ & $O\left(\max\left(m,n\right)\right)$ & \ref{alg:dense} \\
Buffered anti-diagonal including \method{} & $O\left(\left(\max\left(m,n\right)\right)^3\right)$ & $O\left(mn\max\left(m,n\right)\right)$ & $O\left(\max\left(m,n\right)\right)$ & \ref{alg:parallel}, \ref{alg:blocked}, \ref{alg:implicit} \\
Nested sequential sweep & $O\left(\left(\max\left(m,n\right)\right)^3\right)$ & $O\left(mn\max\left(m,n\right)\right)$ & $O\left(mn\right)$ & \ref{alg:nested} \\
\bottomrule
\end{tabular}
\end{scriptsize}
\end{center}
\end{table}

Beyond \cref{alg:parallel,alg:blocked}, two further equivalent algorithms are deferred to the appendices.
\cref{app:nested} gives a sequential nested form that sweeps column by column rather than by anti-diagonal, and \cref{app:implicit} gives implicit-$Y$ realizations that keep a single feedback buffer and reconstruct the corrected entries on demand.
Both produce the same rounded matrix $Z$ and match the $O\left(mn\max\left(m,n\right)\right)$ sweep bound of \method{}.
\cref{tab:complexity} compares the cost of all realizations.

\subsection{Higher-Dimensional Extension}
\label{sec:d-sided}

\begin{algorithm}[!htbp]
\caption{Buffered $D$-sided rounding}
\label{alg:d-sided}
\textbf{Input:} input tensor $X\in\mathbb{R}^{n_1\times\cdots\times n_D}$, feedback factors $L^{\left(1\right)},\ldots,L^{\left(D\right)}$

\textbf{Output:} rounded tensor $Z\in\mathbb{Z}^{n_1\times\cdots\times n_D}$

\begin{algorithmic}[1]
\STATE Initialize $Z\in\mathbb{Z}^{n_1\times\cdots\times n_D}$
\FOR{$d=0,\ldots,D$ \textbf{in parallel}}
    \STATE $C^{\left(d\right)}\gets 0\in\mathbb{R}^{n_1\times\cdots\times n_D}$ \algcomment{buffer}
\ENDFOR
\FOR{$s=D,\ldots,\sum_{d=1}^{D}n_d$}
    \STATE $\bm{\mathcal{I}}_s\gets\left\{\bm{i}\;\middle|\;1\le i_d\le n_d,\ \sum_{d=1}^{D}i_d=s\right\}$
    \FORALL{$\bm{i}\in\bm{\mathcal{I}}_s$ \textbf{in parallel}}
        \STATE $Y_{\bm{i}}\gets X_{\bm{i}}+\sum_{d=1}^{D}C^{\left(d\right)}_{\bm{i}}$\label{line:d-sided-reconstruct} \algcomment{reconstruction}
        \STATE $Z_{\bm{i}}\gets\left\lfloor Y_{\bm{i}}\right\rceil$
        \STATE $C^{\left(0\right)}_{\bm{i}}\gets Z_{\bm{i}}-Y_{\bm{i}}$\label{line:d-sided-error} \algcomment{error}
    \ENDFOR
    \FOR{$d=1,\ldots,D$}
        \FORALL{$\bm{i}\in\bm{\mathcal{I}}_s$ \textbf{in parallel}}
            \STATE $C^{\left(d\right)}_{\bm{i}_{d\mapsto i_d:n_d}}\gets C^{\left(d\right)}_{\bm{i}_{d\mapsto i_d:n_d}}+L^{\left(d\right)}_{i_d:n_d,\,i_d}\,C^{\left(d-1\right)}_{\bm{i}}$\label{line:d-sided-fold}
        \ENDFOR
    \ENDFOR
\ENDFOR
\STATE \textbf{return} $Z$
\end{algorithmic}
\end{algorithm}

The buffered sweep extends to Kronecker-factored objectives in any number of dimensions.
Let $X\in\mathbb{R}^{n_1\times\cdots\times n_D}$ be a tensor with $D$ modes and $N=\prod_{d=1}^{D}n_d$ entries.
Its entries have multi-indices $\bm{i}=\left(i_1,\ldots,i_D\right)$ with $1\le i_d\le n_d$ in each mode $d=1,\ldots,D$.
Let $\bm{i}_{d\mapsto v}$ denote $\bm{i}$ with $i_d$ replaced by $v$.
Similarly, $\bm{i}_{d\mapsto a:b}$ denotes the entries $\bm{i}_{d\mapsto v}$ for $v=a,\ldots,b$.
The entries $\bm{i}_{d\mapsto 1:n_d}$ form the mode-$d$ fiber through $\bm{i}$.
The mode-$d$ product $\times_d$ multiplies an $n_d\times n_d$ matrix and every mode-$d$ fiber of a tensor.
Let $A^{\left(d\right)}\in\mathbb{R}^{n_d\times n_d}$ be a nonsingular basis matrix for each mode $d$.
We round $X$ to an integer tensor $Z\in\mathbb{Z}^{n_1\times\cdots\times n_D}$ under the $D$\nobreakdash-sided objective
\begin{equation}
\left\|\left(Z-X\right)\times_1 A^{\left(1\right)}\cdots\times_D A^{\left(D\right)}\right\|_{\mathrm{F}}^2.
\label{eq:d-sided-objective}
\end{equation}
Define the Gram matrices $G^{\left(d\right)}={A^{\left(d\right)}}^{\top}A^{\left(d\right)}\in\mathbb{R}^{n_d\times n_d}$.
Their inverses have the LDL decompositions $L^{\left(d\right)}\Lambda^{\left(d\right)}{L^{\left(d\right)}}^{\top}=\left(G^{\left(d\right)}\right)^{-1}$, where $L^{\left(d\right)}\in\mathbb{R}^{n_d\times n_d}$ is unit lower triangular and $\Lambda^{\left(d\right)}\in\mathbb{R}^{n_d\times n_d}$ is diagonal.
We vectorize column-major, i.e., $i_1$ runs fastest.
As in the $D=2$ case of \cref{eq:vec,eq:kron-objective,eq:kron-ldl,eq:vectorized-feedback}, the feedback matrix $F\in\mathbb{R}^{N\times N}$ for $\operatorname{vec}\left(X\right)$ is
\begin{equation}
F=L^{\left(D\right)}\otimes\cdots\otimes L^{\left(1\right)}.
\label{eq:d-sided-feedback}
\end{equation}
\cref{alg:d-sided} gives a buffered algorithm for this objective.

\begin{theorem}
\label{thm:d-sided}
\cref{alg:d-sided} produces the same rounded values as \cref{alg:1d} with $x=\operatorname{vec}\left(X\right)$ and $F$ given by \cref{eq:d-sided-feedback}.
\end{theorem}
The proof is in \cref{app:d-sided-proof}.

We treat $D$ as fixed in the complexity bounds.
The sweep costs $O\left(\left(\prod_{d=1}^{D}n_d\right)\left(\max_d n_d\right)\right)$ work in $1+\sum_{d=1}^{D}\left(n_d-1\right)$ stages.
The vectorized sweep costs $O\left(\left(\prod_{d=1}^{D}n_d\right)^2\right)$ work.
Setup costs $O\left(\left(\max_d n_d\right)^3\right)$.

\section{Efficient GPU Implementation}
\label{sec:implementation}

This section optimizes \method{} (\cref{alg:parallel}) for GPUs.
We batch feedback into matrix products and fuse updates within each stage.
\cref{sec:blocked} derives the blocked algorithm, and \cref{sec:kernels} describes its GPU kernels.
\cref{sec:block-widths} selects the block widths, and \cref{sec:performance} evaluates performance.

\subsection{Blocked \method{}}
\label{sec:blocked}

The feedback in \cref{alg:parallel} consists of many row and column updates.
Each update reads and writes up to a whole row or column, with only a few operations per entry.
Such memory-bound updates are far slower than dense matrix products on parallel hardware.
GPTQ owes its practical speed to \emph{lazy block updates}.
It updates only within a block of columns and then applies the block's errors to all remaining columns in one matrix product.
We port this idea to the anti-diagonal sweep, with blocks at any number of levels.

A feedback contribution goes from a rounded source entry to a target entry on a later stage.
It is first read when its target is rounded, so it may be applied at any time after its source is rounded and before its target is.
For stages $s_0\le s\le s_1$, a \emph{flush} applies all contributions from the entries on stages $s_0,\ldots,s$ to the entries on stages $s+1,\ldots,s_1$, as in Lines~\ref{line:flush}-\ref{line:flush-end} of \cref{alg:blocked}.
For each target entry $\left(i,j\right)$, its down product (Line~\ref{line:flush-down}) reads the errors of the sources in column $j$, and its right product (Line~\ref{line:flush-y}) reads the buffer values of the sources in row $i$.
Both products skip indices outside the matrix.
\cref{alg:parallel} rounds each stage $s$, folds the stage's own errors into $C$, and flushes the stage into all later stages.
From here on, the push of stage $s$ means this flush, so a push also updates $C$ below its stage.
A fold then adds only the entry's own error.

\begin{algorithm}[!htbp]
\caption{Blocked \method{}}
\label{alg:blocked}
\textbf{Input:} input matrix $X \in \mathbb{R}^{m\times n}$, feedback factors $L \in \mathbb{R}^{m \times m}$ and $U \in \mathbb{R}^{n \times n}$, block widths $w_1,\ldots,w_h$

\textbf{Output:} rounded matrix $Z \in \mathbb{Z}^{m\times n}$
\begin{algorithmic}[1]
\STATE $w_0\gets 1$
\STATE $w_{h+1}\gets m+n-1$
\STATE Initialize $Z \in \mathbb{Z}^{m\times n}$
\STATE Initialize $E \in \mathbb{R}^{m\times n}$
\STATE $Y \gets X$
\STATE $C \gets 0 \in \mathbb{R}^{m\times n}$
\FOR{$s = 2, \ldots, m+n$}
    \FORALL{$\left(i,j\right) \in \mathcal{I}_s$ \textbf{in parallel}}
        \STATE $Z_{i,\,j} \gets \left\lfloor Y_{i,\,j}\right\rceil$
        \STATE $E_{i,\,j} \gets Z_{i,\,j} - Y_{i,\,j}$
        \STATE $C_{i,\,j} \gets C_{i,\,j} + E_{i,\,j}$ \algcomment{fold own error}
    \ENDFOR
    \STATE $\ell\gets\max\left\{0\le\ell\le h\;\middle|\;\left(s-1\right)\bmod w_\ell=0\right\}$ \algcomment{largest level $\le h$ of a block ending at $s$}
    \STATE $s_0\gets\max_{\ell\le\ell'\le h+1}\left\lfloor\left(s-2\right)/w_{\ell'}\right\rfloor w_{\ell'}+2$ \algcomment{first stage of that block}
    \STATE $s_1\gets\min_{\ell<\ell'\le h+1}\left\lceil\left(s-1\right)/w_{\ell'}\right\rceil w_{\ell'}+1$ \algcomment{last stage of the next larger block}
    \FORALL{$\left(i,j\right)\in\bigcup_{s'=s+1}^{s_1}\mathcal{I}_{s'}$ \textbf{in parallel}}\label{line:flush}
        \STATE $M_{i,\,j}\gets L_{i,\,s_0-j:s-j}\,E_{s_0-j:s-j,\,j}$ \label{line:flush-down}
        \STATE $C_{i,\,j}\gets C_{i,\,j}+M_{i,\,j}$
        \STATE $Y_{i,\,j}\gets Y_{i,\,j}+M_{i,\,j}+C_{i,\,s_0-i:s-i}\,U_{s_0-i:s-i,\,j}$ \label{line:flush-y}
    \ENDFOR\label{line:flush-end}
\ENDFOR
\STATE \textbf{return} $Z$
\end{algorithmic}
\end{algorithm}

\begin{figure}[!htbp]
\begin{center}
\includegraphics{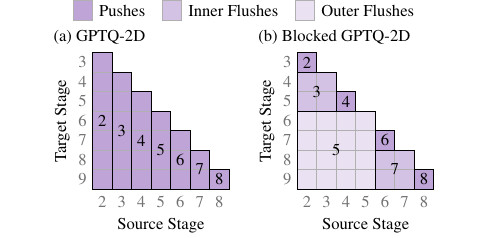}
\end{center}
\caption{Schedules of the feedback between the stages of a $4\times5$ matrix, in (a) \method{} (\cref{alg:parallel}) and (b) blocked \method{} (\cref{alg:blocked}) with inner blocks of $w_1=2$ stages and outer blocks of $w_2=4$. Each rectangle is one flush, labeled with the stage after which it runs.}
\label{fig:flushes}
\end{figure}

Blocked \method{} (\cref{alg:blocked}) groups the stages into blocks at levels $0,\ldots,h+1$.
Level $0$ has blocks of $w_0=1$ stage, and level $h+1$ has a single block of all $w_{h+1}=m+n-1$ stages: $2,\ldots,m+n$.
Each level-$\left(\ell+1\right)$ block splits into level-$\ell$ blocks after every $w_\ell$-th stage of the sweep, with $w_0\le\cdots\le w_{h+1}$.
Panel (c) of \cref{fig:algorithm} shows one stage for $h=2$, where we call the first and second levels of blocks and their flushes inner and outer.
\cref{fig:flushes} instead shows the schedule of a whole sweep, without and with blocks.
Its axes are stages, and each square is the feedback from a source stage to a later target stage.
In both panels, each square lies in the rectangle of exactly one flush, which runs after the square's source stage and before its target stage.
For any blocks, this flush belongs to the largest block that contains the square's source stage but not its target stage.
Each flush changes only entries that are not yet rounded.
So the errors and buffer values that it reads at its sources are final.
Thus \cref{alg:blocked} produces the same $Z$ as \cref{alg:parallel} for all $h$ and widths, and for $h=0$ it is \cref{alg:parallel}.

Blocking changes only when each contribution is applied, so the $m+n-1$ stages and the $O\left(mn\max\left(m,n\right)\right)$ feedback work stay the same.

\subsection{GPU Kernels}
\label{sec:kernels}

We implement \cref{alg:blocked} with Triton kernels.
We intentionally do not fuse the quantizer into these kernels, so that it can be any PyTorch or Triton function.
Each stage calls the quantizer on its anti-diagonal of $Y$ and then launches one push kernel.
The rounded values overwrite their entries of $Y$, so the sweep stores no separate $Z$ and can round $X$ in place.
The whole sweep, tens of thousands of kernel launches for large matrices, can be captured once in a CUDA graph and replayed, which removes the CPU's share of each launch.
In that case, the quantizer must be a function that a CUDA graph can capture.
The experiments in \cref{sec:experiments} use our Triton kernels for the integer quantizers and for QTIP's trellis quantizer.
The trellis kernels returned the same codes as QTIP's implementation in all our tests.

The push kernel fuses the fold with the push's down and right products.
Each target entry belongs to one program, which adds both products to $Y$ in one read-modify-write.
The kernel adds its own errors to the buffer values in registers, for the right products.
A small kernel writes these folds to $C$ once per inner block, before a flush reads them.
The push kernel also stores each stage's errors contiguously.

Before a flush, another small kernel repacks the errors of the flushed stages into dense submatrices, so that the down product of a flush is a dense tiled matrix product.
The right product reads the contiguous row bands of $C$ directly.
Tiles of the rectangular launch grid that lie outside the band of target stages return before loading any operands.
Flushes of wide blocks use larger tiles and launch neighboring tiles together, so that these share their operands in the L2 cache.

Super-entries reuse the same pipeline, including the repack and flush kernels.
Their pushes run as two kernels of small $p\times q$ matrix products, one per product.
The down kernel also writes the folds to $C$.

No kernel uses atomic operations, and each target entry sums its contributions in a fixed order, so every run returns the same bits.
This summation order differs from that of \cref{alg:parallel}, since a flush sums a block's feedback in one product.
Floating-point differences near a quantization boundary can therefore change the rounded value.
In FP64, the rounded matrices equaled those of a PyTorch implementation of \cref{alg:parallel} in all our tests.

\subsection{Block Widths}
\label{sec:block-widths}

\begin{figure}[!htbp]
\begin{center}
\includegraphics{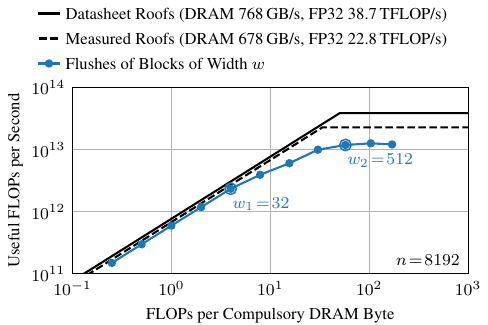}
\end{center}
\caption{Flushes of one level of blocks of width $w$ on a roofline, totaled over a sweep at $n=8192$. Rings mark $w_1=32$ and $w_2=512$.}
\label{fig:roofline}
\end{figure}

We now select the block widths for the NVIDIA RTX A6000 GPU, on which all our experiments run \method{} in FP32 without TF32.
Each target of a flush sums over up to $w$ source stages, where $w$ is the width of the flushed block.
So the flushes do more arithmetic per byte as $w$ grows.
The flushes thus need wide blocks to be compute-bound instead of memory-bound.
We time the flush and push kernels as a sweep launches them, on random $n\times n$ matrices with $n=8192$.
Every measurement is the median of three replays of a CUDA graph.
\cref{fig:roofline} places the flushes of one level of blocks of width $w$ on a roofline.
Its solid roofs are the datasheet peaks.
Its dashed roofs are measured, as the bandwidth of a 1\,GiB device-to-device copy and the FP32 rate of a $4096\times4096$ cuBLAS SGEMM sustained under the 300\,W power cap.
Up to $w=32$, the flushes are memory-bound, and each doubling of $w$ halves their total time.
From $w=512$ on, they run at about 12\,TFLOP/s, half the measured FP32 rate, and wider blocks gain little.
We set $w_2=512$ at this knee.

\begin{figure}[!htbp]
\begin{center}
\includegraphics{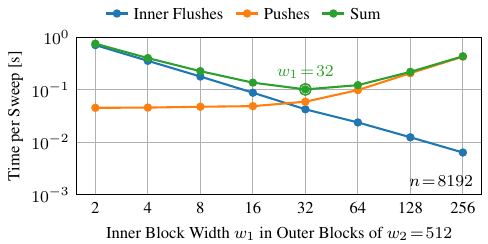}
\end{center}
\caption{Time per sweep of the pushes, the inner flushes, and their sum against $w_1$, for $w_2=512$ at $n=8192$.}
\label{fig:inner-block}
\end{figure}

With a single level of blocks of 512 stages, each push would update the rest of its block, up to 511 later anti-diagonals.
Each entry would thus be read and written once for each earlier stage of its block.
A flush instead reads and writes each target once for all its source stages.
We therefore split each outer block into inner blocks of $w_1$ stages.
An inner flush applies the feedback of each finished inner block to the rest of its outer block.
So a push only needs to update the rest of its inner block, up to $w_1-1$ later anti-diagonals.
\cref{fig:inner-block} shows how $w_1$ trades the pushes against the inner flushes.
Each push is a dependent launch with a small fixed cost, about 2.8\,\textmu s per stage at $n=8192$, even in a CUDA graph.
Up to $w_1=16$, a push updates the rest of its inner block within this cost.
So the pushes take about the same time per sweep, while each doubling of $w_1$ halves the inner flushes.
Beyond that, the pushes grow with their memory traffic.
The sum is lowest at $w_1=32$.

Both widths thus sit where their kernels change bottleneck.
The flushes turn from memory-bound to compute-bound at $w_2$, and the pushes turn from latency-bound to memory-bound at $w_1$.
A third level could mainly shorten the inner flushes, which take about a fifth of a sweep at $n=8192$.
So two levels suffice, and more levels would help little.
Scalar sweeps thus use $w_1=32$ and $w_2=512$.
Sweeps of super-entries use $w_1=8$ and $w_2=64$ stages of super-entries.

\subsection{Performance}
\label{sec:performance}

\begin{figure}[!htbp]
\begin{center}
\includegraphics{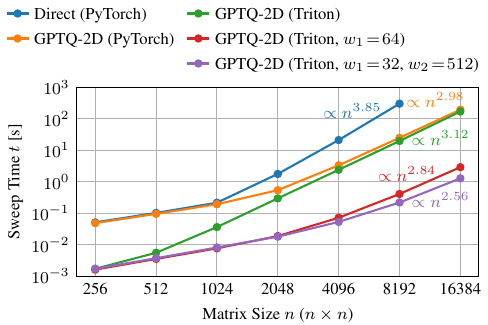}
\end{center}
\caption{Time of one sweep on an $n\times n$ matrix, from \cref{alg:dense} (Direct) to the default two-level blocks. Each label gives the slope of its line's last segment.}
\label{fig:ablation}
\end{figure}

We now present the time $t$ of one sweep that rounds a random $n\times n$ matrix with $\left\lfloor\cdot\right\rceil$, as we add one technique at a time to \cref{alg:dense}.
\cref{fig:ablation} plots $t$ against $n$ on log-log axes, where the slope $\Delta\log t/\Delta\log n$ of a line estimates the asymptotic order of $t$ in $n$.
We do not run \cref{alg:dense} at $n=16384$, where one sweep would take over an hour.
\cref{alg:parallel} in PyTorch is cubic and, at $n=8192$, 12$\times$ faster than the quartic \cref{alg:dense}.
Triton kernels in a CUDA graph gain only 1.3$\times$ more, since the unblocked pushes stay memory-bound.
One level of blocks gains another 49$\times$ with $w_1=64$, the best width for a single level.
The second level ($w_1=32$, $w_2=512$) gains another 1.9$\times$.
At $n=8192$, the default sweep thus takes 219\,ms and is 1{,}400$\times$ faster than \cref{alg:dense}.
At $n=16384$, it runs at 6.8\,TFLOP/s, 30\% of the measured FP32 rate.

\section{Experiments}
\label{sec:experiments}

This section evaluates \method{} against YAQA and GPTQ on four Llama models.
\cref{sec:setup} describes the setup, \cref{sec:accuracy} compares the accuracy, and \cref{sec:timing} compares the quantization time and memory.

\subsection{Setup}
\label{sec:setup}

We quantize four Llama models, under the Kronecker-factored Hessian checkpoints that YAQA estimates with its Sketch B and releases.
We preprocess the Hessians as YAQA does and round each linear layer under the two-sided objective of \cref{eq:intro-objective}.
We use the two formats reported by YAQA's paper.
QTIP 2-bit rounds $16\times16$ super-entries with the 2-bit trellis quantizer of QTIP, after incoherence processing by random Hadamard transforms.
INT4 rounds to GPTQ's symmetric grid $\left\{-8,\ldots,7\right\}$ with a BF16 scale, without incoherence processing.
Each row has one scale for every 32 consecutive input channels, set to $2/15$ of the largest magnitude among their 32 weights.
Neither YAQA's paper nor its released code specifies the INT4 grid.

\subsection{Accuracy}
\label{sec:accuracy}

\begin{table}[!htbp]
\caption{Accuracy of Llama-3.1-8B-Instruct rounded under YAQA's Hessians. KL denotes each model's KL divergence from the BF16 model and PPL its perplexity, both measured on the dataset named in the column head. Rows marked $\dagger$ are reported by YAQA, and rows marked $\ddagger$ are the versions timed in \cref{tab:timing}.}
\label{tab:accuracy}
\begin{center}
\begin{small}
\begin{tabular}{lccc}
\toprule
Method & WikiText2 KL & WikiText2 PPL & C4 PPL \\
\midrule
BF16 & 0 & 6.50 & 8.02 \\
\midrule
\multicolumn{4}{l}{\emph{QTIP 2-bit}} \\
GPTQ$^\dagger$ & 0.356 & 9.39 & 10.70 \\
YAQA$^\dagger$ & 0.241 & 8.39 & 9.83 \\
\method{} & 0.242$\pm$0.002 & 8.40$\pm$0.03 & 9.83$\pm$0.01 \\
\method{}$^\ddagger$ & 0.245$\pm$0.008 & 8.45$\pm$0.10 & 9.84$\pm$0.01 \\
\midrule
\multicolumn{4}{l}{\emph{INT4}} \\
GPTQ$^\dagger$ & 0.033 & 6.75 & 8.21 \\
YAQA$^\dagger$ & 0.021 & 6.65 & 8.15 \\
\method{} & 0.023 & 6.66 & 8.18 \\
\method{}$^\ddagger$ & 0.025 & 6.67 & 8.19 \\
\bottomrule
\end{tabular}
\end{small}
\end{center}
\end{table}

We evaluate the accuracy of quantized Llama-3.1-8B-Instruct in \cref{tab:accuracy}.
Its two unmarked \method{} rows round YAQA's super-entries, $16\times16$ in QTIP 2-bit and $1\times32$ in INT4.
They reverse the rows and columns of the weights and Hessians, so that \method{} rounds from the last row and column, as YAQA does.
In INT4, each scale is computed from the current values of its super-entry.
The rows marked $\ddagger$ are the versions that \cref{tab:timing} times.
They round in forward order and, in INT4, entry by entry with static scales.
We report the KL divergence from the BF16 model on WikiText2 and the perplexities on WikiText2 and C4, with sequences of 8192 tokens.
In QTIP 2-bit, we give the mean and standard deviation over three draws of the random Hadamard signs.
INT4 has no randomness.
For reference, we add the rows marked $\dagger$, which are single runs reported by YAQA.
YAQA reports GPTQ under the name LDLQ.

YAQA's fixed-point iteration is LDLQ under the Kronecker-factored Hessian, and LDLQ is equivalent to GPTQ.
So by \cref{thm:equiv}, \method{} provably computes YAQA's rounding when both use the same order and quantizer.
In our FP64 tests, \method{} returns the same rounded values as YAQA's released rounding code.
\cref{tab:accuracy} shows that the two are also numerically close in FP32.
In QTIP 2-bit, \method{} matches YAQA's reported accuracy within one standard deviation.
In INT4, \method{}'s KL divergence of 0.023 is much closer to YAQA's 0.021 than to GPTQ's 0.033.
YAQA does not specify its INT4 grid, so the remaining gap to YAQA may come from a different grid.

In QTIP 2-bit, \method{} lowers GPTQ's KL divergence by a third.
Rounding in forward order ($\ddagger$) gives the same accuracy within its standard deviations.

\subsection{Quantization Time and Memory}
\label{sec:timing}

\begin{table}[!htbp]
\caption{Time and peak GPU memory to quantize each model on an NVIDIA RTX A6000 GPU, from precomputed Hessians to rounded weights.}
\label{tab:timing}
\begin{center}
\begin{small}
\begin{tabular}{lrrrr}
\toprule
 & \multicolumn{2}{c}{QTIP 2-bit} & \multicolumn{2}{c}{INT4} \\
\cmidrule(lr){2-3}\cmidrule(lr){4-5}
Model & YAQA & \method{} & GPTQ & \method{} \\
\midrule
\multicolumn{5}{l}{\emph{Time [s]}} \\
1B & 4,296 & 127 & 10 & 19 \\
3B & 12,290 & 285 & 17 & 30 \\
8B & 34,711 & 653 & 44 & 73 \\
70B & 740,595 & 5,776 & 533 & 848 \\
\midrule
\multicolumn{5}{l}{\emph{Peak memory [GiB]}} \\
1B & 3.16 & 6.94 & 3.58 & 2.87 \\
3B & 3.34 & 9.75 & 4.05 & 3.87 \\
8B & 12.17 & 15.45 & 9.74 & 7.90 \\
70B & 47.39 & 42.79 & 36.95 & 28.59 \\
\bottomrule
\end{tabular}
\end{small}
\end{center}
\end{table}

We measure the time and peak GPU memory of quantizing Llama-3.2-1B-Instruct, Llama-3.2-3B-Instruct, Llama-3.1-8B-Instruct, and Llama-3.1-70B-Instruct on an NVIDIA RTX A6000 GPU in \cref{tab:timing}.
Each measurement runs from the Hessians in GPU memory to the rounded weights and excludes disk I/O.
The times include incoherence processing, the factorizations, the one-time CUDA graph captures and any runtime compilations.
We run publicly released implementations of both baselines rather than versions of our own.
For GPTQ, we use the efficient implementation of \citet{ICLR2026_c76edb90}, which also applies lazy block updates.
Neither implementation supports the other's format, so we time YAQA only in QTIP 2-bit and GPTQ only in INT4.
In INT4, GPTQ and \method{} both round entry by entry in forward order with static scales per 32 input channels.
All methods quantize the layers sequentially.
We time the first two decoder blocks and count every further decoder block at the second one's time.

\method{} brings the time of YAQA's rounding close to that of standard post-training quantization such as GPTQ.
YAQA's paper expects its rounding to take about twice as long as GPTQ.
This estimate counts only sequential stages, not the total feedback work, which is quartic for YAQA and cubic for GPTQ.
In QTIP 2-bit, \method{} quantizes the 1B model 34$\times$ faster than YAQA and the 70B model 128$\times$ faster.
The speedup grows with the model size since larger models have larger layers.
For the 70B model, this cuts the time from 8.6 days to 1.6 hours.
In INT4, \method{} takes only 1.6$\times$ to 1.9$\times$ the time of GPTQ.
The 70B model takes 14 minutes.

\method{}'s peak memory stays moderate as well.
It is lower than GPTQ's on all four models and lower than YAQA's on the 70B model.
On the smaller models, it exceeds YAQA's because \method{} keeps its captured CUDA graphs.

\section{Conclusion}
\label{sec:conclusion}

Two-sided adaptive rounding keeps a quantized LLM closer to the original model than one-sided rounding, but it has been far more costly to compute.
\method{} computes this rounding exactly, by propagating each rounding error one side at a time.
This lowers the work from quartic to cubic, an asymptotic reduction rather than a constant-factor speedup.
On GPUs, it batches the feedback into dense matrix products.
It thus matches YAQA's accuracy at a cost close to that of GPTQ.
\method{} makes a given rounding faster but not more accurate, so choosing a better two-sided objective is left to future work.

\section*{Acknowledgements}
\addcontentsline{toc}{section}{Acknowledgements}

This research was partially supported by a generous grant from NVIDIA.

\bibliography{bibliography}
\bibliographystyle{arxiv}
\addcontentsline{toc}{section}{References}

\appendix
\crefalias{section}{appendix}
\crefalias{subsection}{appendix}

\section{Proof of \texorpdfstring{\protect\cref{thm:d-sided}}{Theorem 2}}
\label{app:d-sided-proof}

Semicolons separate multiple index replacements, as in $\bm{i}_{r\mapsto u;\,d\mapsto v}$.

\begin{proof}
Since each $L^{\left(d\right)}$ is lower triangular, an entry receives feedback only from entries with no larger index in any mode.
Thus the stages follow a topological order, as in \cref{sec:feedback}.
As in the proof of \cref{thm:equiv}, we compare the feedback at an entry $\bm{i}$ before its stage is rounded.
Let $E=C^{\left(0\right)}$ contain the errors (Line~\ref{line:d-sided-error}) from earlier stages and be zero elsewhere.
Fix modes $1\le r<d\le D$ and an index $1\le v<i_d$.
All sources $\bm{i}_{r\mapsto u;\,d\mapsto v}$ with $1\le u\le i_r$ have folded into $C^{\left(r\right)}_{\bm{i}_{d\mapsto v}}$ (Line~\ref{line:d-sided-fold}).
Each stage folds mode by mode, so these folds read completed source values that later stages cannot change.
Thus,
\begin{equation}
\begin{split}
C^{\left(r\right)}_{\bm{i}_{d\mapsto v}}
&=\sum_{u=1}^{i_r}L^{\left(r\right)}_{i_r,\,u}\,C^{\left(r-1\right)}_{\bm{i}_{r\mapsto u;\,d\mapsto v}}\\
&=\left(C^{\left(r-1\right)}\times_r L^{\left(r\right)}\right)_{\bm{i}_{d\mapsto v}}\\
&=\cdots=\left(E\times_1 L^{\left(1\right)}\cdots\times_r L^{\left(r\right)}\right)_{\bm{i}_{d\mapsto v}}.
\end{split}
\label{eq:d-sided-source-folds}
\end{equation}
Before $\bm{i}$ is rounded, only sources $\bm{i}_{d\mapsto v}$ with $1\le v<i_d$ have folded into $C^{\left(d\right)}_{\bm{i}}$ (Line~\ref{line:d-sided-fold}).
Thus, for $d=1,\ldots,D$,
\begin{equation}
\begin{split}
C^{\left(d\right)}_{\bm{i}}
&=\sum_{v=1}^{i_d-1}L^{\left(d\right)}_{i_d,\,v}\,C^{\left(d-1\right)}_{\bm{i}_{d\mapsto v}}\\
&=\sum_{v=1}^{i_d-1}L^{\left(d\right)}_{i_d,\,v}\left(E\times_1 L^{\left(1\right)}\cdots\times_{d-1}L^{\left(d-1\right)}\right)_{\bm{i}_{d\mapsto v}}\\
&=\left(E\times_1 L^{\left(1\right)}\cdots\times_d L^{\left(d\right)}\right)_{\bm{i}}-\left(E\times_1 L^{\left(1\right)}\cdots\times_{d-1}L^{\left(d-1\right)}\right)_{\bm{i}}.
\end{split}
\label{eq:d-sided-incoming}
\end{equation}
Summing \cref{eq:d-sided-incoming} over the modes telescopes.
With $E_{\bm{i}}=0$, Line~\ref{line:d-sided-reconstruct} gives
\begin{equation}
\begin{split}
Y_{\bm{i}}
&=X_{\bm{i}}+\sum_{d=1}^{D}C^{\left(d\right)}_{\bm{i}}\\
&=X_{\bm{i}}+\left(E\times_1 L^{\left(1\right)}\cdots\times_D L^{\left(D\right)}-E\right)_{\bm{i}}\\
&=X_{\bm{i}}+\left(E\times_1 L^{\left(1\right)}\cdots\times_D L^{\left(D\right)}\right)_{\bm{i}}.
\end{split}
\label{eq:d-sided-reconstruction}
\end{equation}
By \cref{eq:d-sided-feedback}, this is exactly the feedback required by \cref{alg:1d}.
\end{proof}

\section{Other Variants}
\label{app:variants}

\subsection{Nested Algorithm}
\label{app:nested}

Column-major and row-major orders are also topological orders of the dependency graph in \cref{fig:dependency}.
The column-major variant maintains the corrected matrix $Y$ (initialized to $X$) and rounds it column by column.
Column $j$ is rounded top to bottom by the one-dimensional primitive of \cref{sec:1d} under the feedback factor $L$, yielding the rounded-and-corrected column $y$.
The resulting column correction is then propagated to the not-yet-processed columns through the $j$-th row of $U$ by $Y\leftarrow Y+\left(y-Y_{:,\,j}\right)\,U_{j,\,:}$.
Because $U$ is unit upper triangular, this leaves columns $1,\ldots,j-1$ untouched and sets column $j$ to $y$.
\cref{alg:nested} is the column-major form, and the row-major mirror (rounding along rows under $U$, propagating through $L$) is symmetric.
Visiting the entries in column-major order realizes the forward sweep of $\operatorname{vec}\left(X\right)$ itself.
The updates apply the identical rank-one feedback of \cref{eq:rank-one-update}, so the same $Z$ is produced.
Each column costs an $O\left(m^2\right)$ column sweep and an $O\left(m\left(n-j\right)\right)$ propagation, so the total is $O\left(mn\max\left(m,n\right)\right)$ work in $mn$ sequential stages.

\begin{algorithm}[!htbp]
\caption{Nested form (column-major)}
\label{alg:nested}
\textbf{Input:} input matrix $X \in \mathbb{R}^{m\times n}$, feedback factors $L \in \mathbb{R}^{m \times m}$ and $U \in \mathbb{R}^{n \times n}$

\textbf{Output:} rounded matrix $Z \in \mathbb{Z}^{m\times n}$
\begin{algorithmic}[1]
\STATE Initialize $Z \in \mathbb{Z}^{m\times n}$
\STATE $Y \gets X$
\FOR{$j = 1, \ldots, n$}
    \STATE $y \gets Y_{:,\,j}$
    \FOR{$i = 1, \ldots, m$}
        \STATE $Z_{i,\,j} \gets \left\lfloor y_i \right\rceil$
        \STATE $E_{i,\,j} \gets Z_{i,\,j} - y_i$
        \STATE $y_{i:m} \gets y_{i:m} + L_{i:m,\,i} E_{i,\,j}$
    \ENDFOR
    \STATE $Y \gets Y + \left(y - Y_{:,\,j}\right) U_{j,\,:}$
\ENDFOR
\STATE \textbf{return} $Z$
\end{algorithmic}
\end{algorithm}

\subsection{\texorpdfstring{Implicit-$Y$}{Implicit-Y} Realizations}
\label{app:implicit}

\cref{alg:1d,alg:dense,alg:parallel,alg:blocked,alg:nested} maintain the corrected matrix $Y$ explicitly.
An equivalent family instead keeps a single triangular-feedback buffer and reconstructs the corrected entries on demand.
Grouping the rank-one feedbacks of \cref{eq:rank-one-update} along columns accumulates the errors through $L$ into the buffer $C=LE$, so that $Y=X+CU$ and a current entry reads
\begin{equation}
Y_{i,\,j}=X_{i,\,j}+\left(CU\right)_{i,\,j}=X_{i,\,j}+C_{i,\,1:j}U_{1:j,\,j}.
\label{eq:Yij-C}
\end{equation}
After rounding, setting $E_{i,\,j}$ changes only column $j$ of $C$, giving the suffix update
\begin{equation}
C_{i:m,\,j}\leftarrow C_{i:m,\,j}+L_{i:m,\,i}E_{i,\,j}.
\label{eq:C-update}
\end{equation}
Sweeping anti-diagonal by anti-diagonal gives \cref{alg:implicit}.
Grouping along rows instead yields the mirror buffer $EU$, with $Y=X+L\left(EU\right)$.
Both apply the identical feedback of \cref{eq:rank-one-update}, so they produce the same rounded matrix $Z$.
The cost matches that of \cref{alg:parallel}, since reading $C_{i,\,1:j}U_{1:j,\,j}$ costs $O\left(j\right)$ for entry $\left(i,j\right)$ and updating $C_{i:m,\,j}$ costs $O\left(m-i+1\right)$, for a total of $O\left(mn\max\left(m,n\right)\right)$.
\cref{alg:implicit} also takes $m+n-1$ stages, one per anti-diagonal.

\begin{algorithm}[!htbp]
\caption{Implicit-\texorpdfstring{$Y$}{Y} form (anti-diagonal)}
\label{alg:implicit}
\textbf{Input:} input matrix $X \in \mathbb{R}^{m\times n}$, feedback factors $L \in \mathbb{R}^{m \times m}$ and $U \in \mathbb{R}^{n \times n}$

\textbf{Output:} rounded matrix $Z \in \mathbb{Z}^{m\times n}$
\begin{algorithmic}[1]
\STATE Initialize $Z \in \mathbb{Z}^{m\times n}$
\STATE $C \gets 0 \in \mathbb{R}^{m\times n}$
\FOR{$s = 2, \ldots, m + n$}
    \FORALL{$\left(i, j\right) \in \mathcal{I}_s$ \textbf{in parallel}}
        \STATE $Y_{i,\,j} \gets X_{i,\,j} + C_{i,\,1:j} U_{1:j,\,j}$
        \STATE $Z_{i,\,j} \gets \left\lfloor Y_{i,\,j}\right\rceil$
        \STATE $E_{i,\,j} \gets Z_{i,\,j} - Y_{i,\,j}$
    \ENDFOR
    \FORALL{$\left(i, j\right) \in \mathcal{I}_s$ \textbf{in parallel}}
        \STATE $C_{i:m,\,j} \gets C_{i:m,\,j} + L_{i:m,\,i} E_{i,\,j}$
    \ENDFOR
\ENDFOR
\STATE \textbf{return} $Z$
\end{algorithmic}
\end{algorithm}

\end{document}